\def\year{2021}\relax
\documentclass[letterpaper]{article} 
\usepackage{aaai21}  
\usepackage{times}  
\usepackage{helvet} 
\usepackage{courier}  
\usepackage[hyphens]{url}  
\usepackage{graphicx} 
\usepackage{natbib}  
\usepackage{caption} 
\usepackage{color}
\usepackage{subfigure}
\usepackage{amsmath,amssymb,amsfonts}
\usepackage{algorithm}
\usepackage{algorithmic}
\usepackage{amsthm}

\usepackage{epstopdf}

\usepackage{booktabs}

\newtheorem{lemma}{Lemma}[section]
\newtheorem{definition}{Definition}[section]

\title{Enhancing Balanced Graph Edge Partition with Effective Local Search}
\author {
        Zhenyu Guo,\textsuperscript{\rm 1}
        Mingyu Xiao, \textsuperscript{\rm 1\thanks{Corresponding author.}}
        Yi Zhou, \textsuperscript{\rm 1}
        Dongxiang Zhang, \textsuperscript{\rm 2}
        Kian-Lee Tan, \textsuperscript{\rm 3} \\
}
\affiliations {
    \textsuperscript{\rm 1} University of Electronic Science and Technology of China \\
    \textsuperscript{\rm 2} Zhejiang University \\
    \textsuperscript{\rm 3} National University of Singapore \\
    Harry.Guo@outlook.com, myxiao@gmail.com, zhouyi0922@gmail.com, zhangdongxiang@zju.edu.cn, tankl@comp.nus.edu.sg \\
}

\begin{document}

\maketitle

\begin{abstract}
    Graph partition is a key component to achieve workload balance and reduce job completion time in parallel graph processing systems. Among the various partition strategies, edge partition has demonstrated more promising performance in power-law graphs than vertex partition and thereby has been more widely adopted as the default partition strategy by existing graph systems. The graph edge partition problem, which is to split the edge set into multiple balanced parts to minimize the total number of copied vertices, has been widely studied from the view of optimization and algorithms. In this paper, we study local search algorithms for this problem to further improve the partition results from existing methods. More specifically, we propose two novel concepts, namely adjustable edges and blocks. Based on these, we develop a greedy heuristic as well as an improved search algorithm utilizing the property of the max-flow model. To evaluate the performance of our algorithms, we first provide adequate theoretical analysis in terms of the approximation quality.
    We significantly improve the previously known approximation ratio for this problem. Then we conduct extensive experiments on a large number of benchmark datasets and state-of-the-art edge partition strategies. The results show that our proposed local search framework can further improve the quality of graph partition by a wide margin.
\end{abstract}

\section{Introduction}
Graph partition plays a key role in performance improvement for massive graph processing systems. In a distributed graph system, such as Google Pregel \cite{malewicz2010pregel}, GraphX \cite{gonzalez2014graphx}, and GraphLab \cite{low2012distributed}, the original graph may be too large to fit in memory and has to be partitioned into multiple parts which are processed in parallel by multiple machines. 
The quality of graph partition is often measured by two important performance criteria. One is \textit{workload balance} which expects the sizes of the partitioned parts to be as equal as possible. The goal is to reduce the overall job completion time (JCT) in parallel systems, where the bottleneck is caused by the slowest job. The other is \emph{communication overhead} whose objective is to minimize the connection among the parts. It is challenging to compute the optimal partition as the problem has been proven to be NP-Hard \cite{goemans1995improved,feder1999complexity}.

Vertex partition is a popular model in which the workload of each part is evaluated by its number of vertices and the communication overhead between two parts is evaluated by their connecting edges in the original graph.
In the past decades, there have been significant efforts devoted to this problem, including both theoretical results~\cite{andreev2006balanced,feldmann2013fast} and heuristic algorithms.
However, the performance of vertex partition models may degrade for parallel algorithms to handle power-law graphs \cite{gonzalez2012powergraph}.
In fact, most natural graphs follow a skewed degree distribution similar to the power-law distribution \cite{faloutsos1999power,newman2001random}.
We can observe that in real-world graphs, a small fraction of vertices may connect to a large part of the graph. For example, celebrities in a social network attract a huge number of followers. This property brings non-trivial challenges to vertex partitioners, including workload balance, partitioning, communication, storage, and computation \cite{gonzalez2012powergraph}.

To address the issue of power-law distribution, edge partition was introduced to partition the graph based on edge sets \cite{gonzalez2012powergraph}.
A vertex is allowed to appear in multiple parts sharing this vertex.
The workload of edge partition is evaluated by the number of edges in each part while the communication overhead is evaluated by the \emph{replication factor}, which indicates the average time of each vertex appearing in all the parts. Edge partition is more efficient in power-law graphs, and
several parallel graph processing systems, including  PowerGraph \cite{gonzalez2012powergraph}, Spark GraphX \cite{gonzalez2014graphx} and Chaos \cite{roy2015chaos}, have adopted edge partition as the default partition strategy.

\subsection{Contributions}
In this paper, we focus on edge partition.
 Instead of proposing a new edge partition method, we adopt local search techniques to further improve the partition results from existing methods. We will prove several structural properties of edge partition and introduce two novel concepts named \emph{adjustable edges} and \emph{blocks}. Based on these concepts, we develop two types of effective local search algorithms, namely LS-G and LS-F, which are complementary to state-of-the-art approaches. From the initial partition solutions derived from existing partition methods, our local search algorithms can further improve the solution by neighborhood operators. LS-G is a fast greedy heuristic and LS-F leverages the max-flow model to yield partitions with higher quality.
 In theory, we present theoretical analysis in terms of the approximation quality and provide improved approximation ratios for this problem.
 In practice, our experiments are conducted on multiple large-scale benchmark datasets and results show that the solutions derived from state-of-the-art edge partition methods can be further improved by a wide margin.

\section{Preliminaries}

Let $G=(V,E)$ stand for an undirected graph with $n=|V|$ vertices and $m=|E|$ edges.
The neighbor set of a vertex $v$ is denoted by $N(v)=\{u\mid{u,v}\in E\}$ and the neighbor set of a vertex set $S$ is denoted by $N(S)$.
For a subgraph or an edge set $G'$, we use $V(G')$ to denote the set of vertices appearing in $G'$ and $E(G')$ to denote the set
of edges appearing in $G'$.
For an edge subset $E'\subseteq E$, we use $G[E']$ to denote the subgraph \emph{induced} from the edge set $E'$,
i.e., the graph $(V(E'), E')$.

Given a graph $G=(V,E)$, a \emph{$k$-edge partition} divides $E$ into $k$ disjoint groups,
denoted by $P=\{E_1, E_2, \cdots, E_k\}$,  where $E_i \cap E_j = \emptyset, \forall i\ne j$ and $\bigcup_{1\le i \le k} E_i = E$.
A $k$-edge partition $P$ is \emph{$\alpha$-balanced} if each part $E_i$ in $P$ satisfies:
$$|E_i|\le \left\lceil \alpha \frac{|E|}{k} \right\rceil.$$
The \emph{replication factor} of a $k$-edge partition $P$, denoted by  $RF(P)$, is defined as follows:
$$RF(P)=\frac{1}{|V|}\sum_{i=1}^k |V(E_i)|.$$

Given a graph $G=(V,E)$ and two constants $k$ and  $\alpha$,  the \textsc{Edge Partition Problem} (\textsc{EPP}) is to find an $\alpha$-balanced $k$-edge partition $P$ such that the replication factor $RF(P)$ is minimized.


\section{Adjustable Edges and Blocks}
In this section, we first introduce some basic structural concepts that will be used in our local search strategies.
As mentioned above, our algorithms are local-search algorithms based on a given $k$-edge partition.
So next, we always assume that a $k$-edge partition $P=\{E_1, E_2, \cdots, E_k\}$ is given, which can be obtained by known algorithms or
a random assignment.
Based on a given $k$-edge partition, we will move some edges from one part to other parts to decrease the communication load (replication factor), and in the meanwhile keep each part under the workload balance constraint.

For the \textsc{Vertex Partition Problem}, a local-search strategy is easy to develop. We can move a vertex subset from one part to another part as long as the receiving part is still under the required workload balance and the number of crossing edges among the parts can be decreased.
However, the local search strategy in \textsc{Edge Partition Problem} is more complicated because to decrease the replication factor, we are often required to move a subset of edges from one part to \textit{multiple} different parts simultaneously. In this paper, we will present novel and effective strategies for edge movements. Before that, we first introduce the concept of \textit{adjustable edge} which is important to understand our local search algorithms.

\begin{definition}[\textbf{Reachability}]
Let $P=\{E_1, E_2, \cdots, E_k\}$ be a $k$-edge partition. A part $E_i$ is \emph{reachable} for an edge $(u,v)$ if $E_i$ contains both of the two endpoints of the edge, i.e., $u,v \in V(E_i)$.
\end{definition}

\begin{definition}[\textbf{Adjustable Edge}]
Let $P=\{E_1, E_2, \cdots, E_k\}$ be a $k$-edge partition. An edge $(u,v)\in E_i$ is called an \emph{adjustable edge} if there exists $E_j$ such that $j\neq i$ and its two endpoints $u,v\in V(E_j)$.
\end{definition}

For an adjustable edge, we may be able to move it to another reachable part without increasing the replication factor. A simple approach is to treat the movement of the adjustable edge as a possible way to find better solution. .
However, the replication factor can only be strictly decreased when the adjustable edge has a degree-$1$ endpoint in $G[E_i]$. In this case, the degree-1 endpoint will disappear from $G[E_i]$ after removing the edge, and thus reduce the total number of vertices $\sum_{i=1}^{k}|V(E_i)|$. Unfortunately, the number of degree-$1$ endpoints is limited, which restricts the optimization space for local search.
Our strategy is to find further cases in which the replication factor can be strictly decreased by moving adjustable edges to other reachable parts. Based on the idea, we present another key structure called \emph{block}.

\begin{definition}[\textbf{Block}]
Let $P=\{E_1, E_2, \cdots, E_k\}$ be  a $k$-edge partition and $E^*$ be the set of adjustable edges.
Each connected component in the subgraph $G[E_i]-E^*$ is called a \emph{block} of $E_i$. A block consists of a single vertex
is called a \emph{vertex block}.
\end{definition}

Our local search is designed to move a block from one part to another part. In detail, we first move all the adjustable edges incident on this block to other reachable parts. Note that this step will not increase the replication factor. After that, if we move a block $C$ from part $E_i$ to another part $E_j$, we can decrease the replication factor by
$$\frac{1}{|V|}|V(C)\cap V(E_j)|.$$

Figure \ref{bloExp} shows an example of block movement that can reduce the replication factor. Consider a block $C$ in part $E_1$ connecting $3$ adjustable edges $e_1,e_2,e_3$. In the first step, we move $e_1$ and $e_3$ to part $E_2$, $e_2$ to part $E_3$, and the replication factor remains the same. In the next step, we move the whole block $C$ from part $E_1$ to part $E_4$. Since $V(E_1)\cap V(E_4)=\{v_1,v_2\}$, we can reduce $2$ duplicate vertices.

\begin{figure}[ht]
\centering
\subfigure[Before moving block $C$]{
    \centering
    \includegraphics[width=0.8\linewidth]{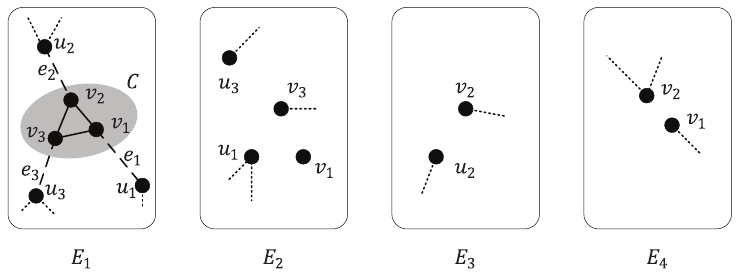}
    \label{bloExp:a}
}
\subfigure[After moving block $C$]{
    \centering
    \includegraphics[width=0.8\linewidth]{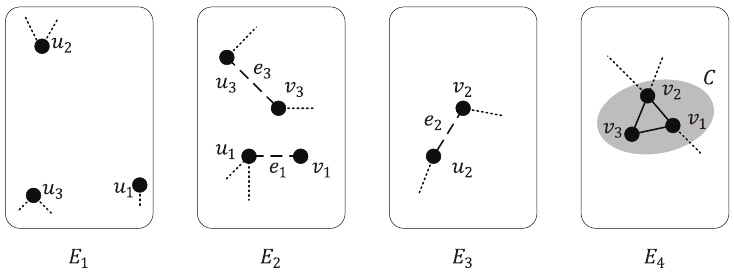}
    \label{bloExp:b}
}

\caption{An example of block movement, where the black part containing three vertices $v_1, v_2$ and $v_3$ is the block $C$, solid edges mean edges in the block, dashed edges mean adjustable edges, and dotted edges mean other edges.}\label{bloExp}
\end{figure}

Different ways to move adjustable edges and blocks will generate different algorithms. We need to consider how and when to move
adjustable edges and blocks under the balance constraint. In this paper, we will give two ideas to do these adjustable operations,
one is based on a greedy idea and one is based on the max flow technique.

\section{A Fast Algorithm: LS-G}
In this section, we introduce a simple and fast local-search algorithm.
We first introduce the most important ingredient of the algorithm, which is the sub-algorithm for the adjustment operation for a single block. Then we present the complete algorithm.

\noindent\textbf{Greedy Adjustments.}~
Let $C$ be a block in part $E_i$. We use $A(C)$ to denote the set of adjustable edges incident on $C$ in the subgraph $G[E_i]$,
i.e., $A(C)=\{(a,b)|(a,b)\in E_i \mbox{~is an adjustable edge}~\land~|\{a,b\}\cap V(C)|=1 \}$. For a block $C$ in $E_i$, the algorithm first checks whether $C$ can be moved to another part $E_j\neq E_i$ to decrease the replication factor. After successfully moving $C$, then we consider the adjustable edges in $A(C)$ in random order. For each adjustable edge $e\in A(C)$, we check whether it can be moved to a reachable part different from $E_i$.
If any step of the algorithm cannot be executed, then we undo all the moving operations in the algorithm.
The algorithm to check whether a block $C$ in  $E_i$ can be moved to another part is denoted by
RA$(C)$ and its pseudocode is shown in Algorithm \ref{RA}.

In an adjustment operation, we need to move all the adjustable edges in $A(C)$ to other parts and then move the total block to another part to decrease the replication factor. For the special case that the block is a vertex block, the second step of moving the block can be omitted since the vertex will automatically disappear in $E_i$ after removing all the adjustable edges incident on it out of $E_i$. In our implementation, we first move the whole block $C$ (the edges in $C$) to another part $E_j$ and then consider moving the adjustable edges in $A(C)$.
It may be more effective to check the feasibility: when $C$ contains several edges, it may be hard to find a part $E_j$ to ``receive'' all the edges together under the balance constraints, while the edges in $A(C)$ can be moved to different parts separately.
\begin{algorithm}[htb]
\caption{ RA$(C)$}\label{RA}
\begin{algorithmic}[1]
 \IF {there are some $E_j\neq E_i$ such that $|E_j|\leq \lceil \alpha \frac{|E|}{k} \rceil -|V(C)|$ and $V(E_j)\cap V(C)\neq \emptyset$}
       \STATE {Let $E_{j'}$ be a part satisfying the condition such that $|V(E_{j'})\cap V(C)|$ is maximized}
       \STATE {$E_i\leftarrow E_i\setminus E(C)$ and $E_{j'}\leftarrow E_{j'}\cap E(C)$}
       \ELSE {}
       \STATE {goto Step 15}
       \ENDIF
\FOR {each adjustable edge $e\in A(C)$}
       \IF {there is a reachable part $E_j\neq E_i$ such that $|E_j|< \lceil \alpha \frac{|E|}{k} \rceil$}
       \STATE {$E_i\leftarrow E_i\setminus\{e\}$ and $E_j\leftarrow E_j\cap\{e\}$}
       \ELSE {}
       \STATE {goto Step 15}
       \ENDIF
     \ENDFOR
\STATE {\textbf{return} yes}
\STATE {Undo all the moving operations and \textbf{return} no}
\end{algorithmic}
\end{algorithm}

\noindent\textbf{The Algorithm LS-G.}~
The whole algorithm, named LS-G, is shown in Algorithm \ref{LS-G}. It scans the parts from an initial solution of any existing edge partition method. For each part $E_i$, we identify the adjustable edges as well as the blocks. Then, we apply RA$(C)$ to deal with each block as the adjustment operation. The blocks within a part are processed in non-decreasing order of their sizes. Note that edge movement will not trigger block reconstructions based on the following lemma.

\begin{algorithm}[htb]
    \caption{LS-G}\label{LS-G}
    \begin{algorithmic}[1]
     \REQUIRE a graph $G=(V,E)$ with an $\alpha$-balanced $k$-edge partition $P=\{E_1,\dots, E_k\}$ of the edge set $E$.
     \ENSURE an $\alpha$-balanced $k$-edge partition.
     \STATE {all parts of $P$ are marked `b'}
     \WHILE {there is a part $E_i$ marked `b'}
     \STATE {mark $E_i$ with `w'}
     \STATE {compute and order all blocks $C_{i1},C_{i2}, \dots, C_{il_i}$ in $E_i$}
    \FOR{$j=1$ \TO $l_i$}
        \STATE call the algorithm RA$(C_{ij})$
        \STATE for any changed part during RA$(C_{ij})$, mark it `b'
    \ENDFOR
      \ENDWHILE
      \RETURN {$P=\{E_1,\dots, E_k\}$}
    \end{algorithmic}
    \end{algorithm}



\begin{lemma}\label{lem_one}
Given a $k$-edge partition $P=\{E_1, E_2, \cdots, E_k\}$ of the edge set of graph $G=(V,E)$.
Let $C$ be a block in part $E_i\in P$ and $e\in E_i$ be an edge not in $C$. Let $P'=\{E'_1, E'_2, \cdots, E'_k\}$ be the new edge partition obtained by moving $e$ to another part $E_j\neq E_i$, where $E'_l=E_l$ for $l\in\{1,2,\dots, k\}\setminus \{i,j\}$.
Then $C$ is still a block in $E'_i$ of the new partition $P'$.
\end{lemma}
\begin{proof}
Since we only move an edge in $E_i\setminus E(C)$ to another part $E_j$, we know that $C$ is still a connected subgraph in $G[E_i]$.
For any edge $(a,b)$ in $E(C)\cap E'_i\subseteq E(C)\cap E_i$,
it is an adjustable edge in $P$ before moving $e$ and there is a part $E_{i_0}\neq E_i$ such that
$a,b \in V(E_{i_0})$. Note that after moving $e$ to $E_j$, no matter if $E_{i_0}=E_j$ or not, no vertex in $V(E_{i_0})$ will be removed,
i.e., $V(E'_{i_0})\subseteq V(E_{i_0})$. So part $E_{i_0}$ is still a reachable part for $(a,b)$ and then $(a,b)$
is still an adjustable edge. Thus, $C$ is still a block in $E'_i$ of the new partition $P'$.
\end{proof}
Lemma~\ref{lem_one} implies that  after moving $C$ together with $A(C)$ from a part $E_i$ to other parts, any other block in $E_i$  is
still a block in the new edge partition.
We do not need to compute new blocks in $E_i$ after the operators.

\section{An Algorithm Based on Max Flow: LS-F}
In this section, we introduce an algorithm with a more sophisticated technique for adjustment operation based on blocks.
The aforementioned algorithm RA$(C)$ is fast. However, as it heuristically moves blocks one at a time, it may fail to move certain blocks (due to the constraints).
To search for more domains, we suggest the following adjustment algorithm based on the max-flow model.
 This algorithm considers several different blocks together in each iteration to find better movements for adjustable edges.
To ensure that we can move several blocks simultaneously, we need the following definitions.


\medskip
\noindent\textbf{Independent Block Set.}~
The strategy of our algorithm recommends us to seek for some blocks which do not affect each other when moving together.
Based on this motivation, we present important structures about blocks called \emph{independent block set} as our adjustment operation. We will move all blocks in an independent block set simultaneously.

 \begin{definition}[Independent Block Set]\label{indepBlockDef}
    Let $P$ be a $k$-edge partition of a graph $G$. Two blocks $C_i$ and $C_j$ in $P$ are called \emph{independent} if they are from the same part of $P$ or the shortest distance between $V(C_i)$ and $V(C_j)$ in $G$ is at least two.
    A set of blocks $\mathcal{C}=\{C_1,C_2,\dots,C_l\}$ is \emph{independent} if
    any pair of blocks in it are independent.
\end{definition}

In the above definition, we set the distance between two blocks is at least two. So no two blocks in an independent set intersect, and the adjustable edges incident on two blocks in an independent set are different.
Thus, for a set of independent blocks, we can move all the adjustable edges incident on all blocks in the set to other parts simultaneously without increasing the replication factor. After this, we may reduce the replication factor by moving these
blocks.

\medskip
\noindent\textbf{Adjustments Based on Max Flow.}~
Given an independent block set $\mathcal{C}$, 
we consider whether we can move together all adjustable edges in $\cup_{C\in \mathcal{C}}A(C)$ from its own part to other reachable parts under the balance constraints.
Since we consider moving all adjustable edges incident on a set of blocks together, we may be able to reach more search domains and find better results.
This is the advantage of this method, compared with the previous greedy method which only considers one block each time.

We use a max flow model to solve the problem of moving adjustable edges incident on a set of blocks together under the
balance constraints.
We construct a directed graph $H=(V_H, A_H)$ and reduce our problem to the problem of finding a maximum flow in $H$ from the source $v_{source}$ to the sink $v_{sink}$. The graph $H=(V_H, A_H)$ is constructed as follows,
where $V_H=V_{edge}\cup V_{part} \cup \{v_{source}, v_{sink}\}$.

\begin{itemize}
\item Introduce two vertices, the source $v_{source}$ and the sink $v_{sink}$.
\item For each adjustable edge $e\in\bigcup_{C\in \mathcal{C}}A(C)$, introduce a vertex $v_e$ with an arc $\overrightarrow{v_{source}v_e}$ from the source $v_{source}$ to $v_e$ of capacity $c(v_{source}v_e)=1$. The set of vertices corresponding to edges in $\bigcup_{C\in \mathcal{C}}A(C)$ is denoted by $V_{edge}$.
\item For each part $E_i$ of $P$, introduce a vertex $v_{E_i}$ with an arc $\overrightarrow{v_{E_i}v_{sink}}$ from $v_{E_i}$ to the sink $v_{sink}$ of capacity $c(v_{E_i} v_{sink}) = \Delta_i$, where $\Delta_i$ is the remaining capacity for part $E_i$ to reach the bound of the balance constraint.
\item For each vertex $v_e\in V_{edge}$, add an arc $\overrightarrow{v_ev_{E_i}}$ from $v_e$ to $v_{E_i}$ of capacity $c(v_ev_{E_i})=1$ for each reachable part $E_i$ of $e$ except the original part containing $e$.
\end{itemize}

In the above model, we have not given the precise definition of $\Delta_i$ but it does not cause trouble in understanding the model.
In fact, we will let $\Delta_i=\lceil \alpha \frac{|E|}{k} \rceil - |E^*_i|$.
Here we use $E^*_i$ instead of $E_i$ because
we still need to save some space for moving blocks.
Assume that two blocks $C_1$ and $C_2\in \mathcal{C}$ will be moved to $E_i$, then we will let $E^*_i=E_i\cup C_1\cup C_2$.
However, we no longer know which block will be moved to which part.
To fix this and simplify the algorithm, in our algorithm, we will first determine the 'destination part' for each block before moving
the adjustable edges. We will select the destination part as the part after moving the block to where the replication factor is minimized.

Let $f$ be a maximum flow in $H$, which can be computed by standard max-flow algorithms.
For a block $C\in \mathcal{C}$, if for any edge  $e\in A(C)$ it holds that $f(\overrightarrow{v_{source}v_e})=1$, i.e., there
is an individual flow going through the arc $\overrightarrow{v_{source}v_e}$ in $f$, then we let the indication function $I(C)=1$;
otherwise let $I(C)=0$. Let $\mathcal{C}'$ be the set of blocks with $I(C)=1$.

We claim that based on the flow $f$ we can move all adjustable edges in $A(C)$ for all blocks in $\mathcal{C}'$ to other parts without breaking the balance constraints.
For each block $C\in\mathcal{C}'$, the algorithm moves each adjustable edge $e\in A(C)$
to part $E_j$ for $f(\overrightarrow{v_ev_{E_i}})=1$, i.e., there is an individual flow going through the arc $\overrightarrow{v_ev_{E_i}}$ in $f$.
In fact, for a vertex $v_{E_i}\in V_{part}$, the flow on the arc $\overrightarrow{v_{E_i}v_{sink}}$ is at most $\Delta_i$, thus the number of adjustable edges moving to $E_i$ will not break the balance constraint even no edge moves out from $E_i$ and all edges in $E_i^*\setminus E_i$ move to $E_i$.

We will use MF$(\mathcal{C})$ to denote the above algorithm based on the maximum flow in $H$.
Note that MF$(\mathcal{C})$ only moves edges in $A(C)$ for blocks $C$ with $I(C)=1$ and keeps unchanged for blocks $C$ with $I(C)=0$. In fact, we will also undo the moving of blocks $C$ with $I(C)=0$.

\medskip
\noindent\textbf{Algorithm LS-F.}~
The algorithm is to iteratively deal with an independent block set by
calling MF$(\mathcal{C})$. Two concerns are mainly resolved: the generation of independent sets of blocks and the algorithm termination condition.
We do not need to find a maximum independent set of blocks because it is hard to compute and not very helpful in our heuristic algorithm.
In our algorithm, we find an independent set of blocks containing at most one block from each part by a greedy method.
We first pick an arbitrary block in the first part, and then iteratively try to pick a block from the next part that is independent with all
picked blocks.
Note that there may be many different independent block sets and it is time-consuming to consider all of them.
In fact, it will be a rare case that we can not find an independent block set of size at least two after a large number of iterations.
So we set the stop condition of our algorithm as a running time bound or a maximum number of rounds.

\section{Approximation Ratio}
We first provide theoretical analysis for \textsc{Edge Partition Problem} in terms of the approximation quality.
The previous known approximation ratio for this problem is $O(d_{max}\sqrt{\log k\log n})$, which was first proved on graphs with some restrictions in~\cite{bourse2014balanced} and then extended to general graph in \cite{li2017simple}.
In this section, we will show an approximation ratio of $\min\{k, \widetilde{d}\}$, where $\widetilde{d}$ is the average degree of the graph. Note that $\widetilde{d}\leq d_{max}$. The new result significantly improves the previous approximation ratio.
We will also consider the lower bounds on the approximation ratio of our algorithms.

\begin{lemma}\label{d_average}
    Any feasible edge partition $P$ of a graph $G=(V,E)$ is an approximation solution with ratio at most $\min\{k, \widetilde{d}\}$.
\end{lemma}

\begin{proof}
Each vertex can be copied at most $k$ times in any feasible edge partition since there are only $k$ parts.
So it is trivially to get the approximation ratio of $k$.

Each edge appears only once in a feasible edge partition and then it can contribute at most $2$ vertices. So
it always holds that
$\sum_{i=1}^k |V(E_i)|\leq 2|E|$.
Thus,
    $$RF(P) \le \frac{2|E|}{|V|}=\widetilde{d}.$$
On the other hand, the optimal replication factor can be $1$ (when no vertex is copied).

Thus we get a bound for the approximation ratio $\min\{k, \widetilde{d}\}$.
\end{proof}

Note that the result in the above lemma does not rely on any algorithms. Any feasible solution will hold the approximation ratio.
Next, we consider the approximation ratio related to our algorithms.

Our algorithms will try to move adjustable edges and blocks to decrease the replication factor.
We show that when the edge partition does not have any adjustable edges the approximation ratio $k$ in Lemma~\ref{d_average} can be slightly improved.

\begin{lemma}\label{noAE}
  For a feasible edge partition $P$ of a graph $G=(V,E)$,
  if $P$ does not have any adjustable edges, then it is an approximation solution with ratio $\leq \min\left\{\widetilde{d}, t(\frac{|I_t|}{|V|} + 1)\right\}$,
  where $t=\frac{k+1}{2}$, $V_t = \{v|v\in V \land |N(v)|> t\}$ and $I_t$ is a maximum independent set in the induced graph $G[V_t]$.
\end{lemma}

\begin{proof}
By Lemma~\ref{d_average}, we know that the approximation ratio is always not greater than $\widetilde{d}$. We only need to consider
$t(\frac{|I_t|}{|V|} + 1)$.

    Let $P=\{E_1, E_2, \dots, E_k\}$ be a feasible edge partition having no adjustable edges.
    For a vertex $v$, we use $p_v$ to denote the number of parts in $P$ containing $v$.
    Let $V_{>} = \{v\in V | p_v>t\}$ and $V_{\le }=V\setminus V_{>}$.

    We can see that $V_{>}$ is an independent set, otherwise, there are two vertices $u,v\in V_{>}$ such that $p_u+p_v>2t=k+1$, which
implies that $u$ and $v$ will appear in at least two same parts and then edge $(u,v)$ would be adjustable.
Since $V_{>}$ is an independent set, we have that $|V_{>}|\le |I_t|$.

So it holds that
    $$\begin{aligned}
        RF(P)&=\frac{1}{|V|}\left(\sum_{v\in V_{>}}p_v+\sum_{v\in V_{\le}}p_v\right)\\
        &\le\frac{1}{|V|}\left(k|V_{>}|+\frac{k+1}{2}|V_{\le}|\right)\\
        &<\frac{1}{|V|}\left((k+1)|I_t|+\frac{k+1}{2}(|V|-|I_t|)\right)\\
        &=t\left(\frac{|I_t|}{|V|} + 1\right).
    \end{aligned}$$
\end{proof}

Note that when the graph is sparse, the number of vertices in $V_t$ may not be large, and then the maximum independent set in $G[V_t]$ will be small. For this case, the ratio $t(\frac{|I_t|}{|V|} + 1)$ will be strictly smaller than $k$.

The condition in Lemma~\ref{noAE} is not easy to achieve. But at least Lemma \ref{noAE} implies that the solution quality may be better when there are
fewer adjustable edges.

On the other hand, we show that the approximation ratio $\widetilde{d}$ cannot be improved even when there are no adjustable edges.
We give an example of the $1$-balanced $k$-partition problem ($k=p(p-1)/2$ for some integer $p$), where the ratio $\widetilde{d}$ holds.
The graph contains $k$ independent cliques, each of which has exactly $k$ edges.
In the optimal solution, each connected component is partitioned to one part and then each vertex appears in exactly one part.
We can also construct a solution, where each part takes exactly one edge from each connected component. Then each vertex appears
in exactly  $(p-1)/2=\widetilde{d}$ parts. Furthermore, in this partition, there is no adjustable edge.
Please see Figure \ref{ApproTight} for an illustration of $k=3$.

\begin{figure}[ht]
    \centering
     \setlength{\belowcaptionskip}{-0.5cm}  
    \subfigure[An optimal partition]{
        \centering
        \includegraphics[width=0.7\linewidth]{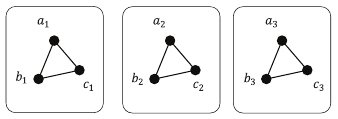}
        \label{ApproTight:a}
    }
    \subfigure[A partition without adjustable edges]{
        \centering
        \includegraphics[width=0.7\linewidth]{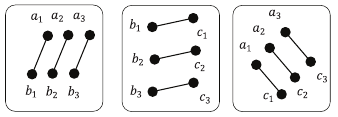}
        \label{ApproTight:b}
    }

    \caption{An example to achieve the ratio $\widetilde{d}$.}\label{ApproTight}
\end{figure}

We then show the approximation ratio $O(k)$ is also tight when no adjustable edges exist.
Note, although the approximation ratio $t(\frac{|I_t|}{|V|} + 1)$ we proved above is slightly better than $k$, this ratio still belongs to $O(k)$ since $|I_t|$ may be as large as $|V|$.
We still use an example of the $1$-balanced $k$-partition to illustrate this tight result.
The graph consists of $k$ independent complete bipartite subgraphs $K^1_{k^2, k}, K^2_{k^2,k}, \dots, K^k_{k^2, k}$, where each subgraph $K^i_{k^2,k}$ has exactly $k^2$ vertices on one side denoted by $U^i=\{u^i_1,u^i_2,\dots, u^i_{k^2}\}$ and $k$ vertices on the other side denoted by $V^i=\{v^i_1,v^i_2,\dots,v^i_{k}\}$.
Similar to the previous case, in the optimal solution, each connected component is partitioned into one part and then each vertex appears in exactly one part.
We can also construct a solution without any adjustable edges: for each subgraph $K^i_{k^2,k}$, we pick all incident edges of $v^i_j$ to part $j$.
Then each vertex in $V^i$ appears in $k$ parts while each vertex in $U^i$ appears in $1$ part.
Totally, the replication factor of this partition is $\frac{k^2+1}{k+1}=O(k)$.
Please see Fig.~\ref{ApproTight2} for an illustration of $k=2$.

\begin{figure}[ht]
    \centering
    \setlength{\belowcaptionskip}{-0.5cm}   
    \subfigure[An optimal partition]{
        \centering
        \includegraphics[width=0.65\linewidth]{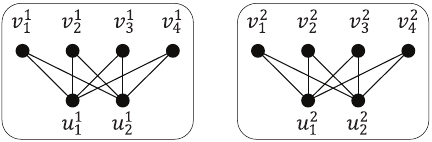}
        \label{ApproTight2:a}
    }
    \subfigure[A partition without adjustable edges]{
        \centering
        \includegraphics[width=0.65\linewidth]{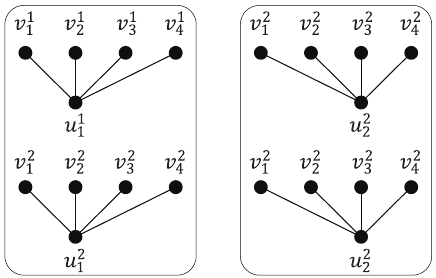}
        \label{ApproTight2:b}
    }

    \caption{An example to achieve the ratio $O(k)$.}\label{ApproTight2}
\end{figure}

\section{Computational Experiments}
In this section, we evaluate the performance of our proposed algorithms LS-G and LS-F. Our objective is to minimize the replication factor (RF), which will be regarded as the quality measure.

\noindent{\textbf{Comparing Algorithms.}}
We consider initial edge partitions generated by the following four algorithms:
METIS \cite{karypis1998fast}, NE \cite{zhang2017graph}, SHEEP \cite{margo2015scalable} and SPAC~\cite{li2017simple}.
The very recent dSPAC~\cite{sebastian2019Scalable} is only a parallelized version of SPAC, and it generates the same edge partition as SPAC. So we do not need to consider it.
We will show the results obtained by our LS-G and LS-F with initial edge partitions generated by them. The algorithm after running our local search algorithm A on the initial partition generated by B is denoted by B+A, where A can be LS-G or LS-F, and B can be METIS, NE, SHEEP, or SPAC.


\noindent{\textbf{Environment.}} Our algorithms are implemented in C++\footnote{Our code is put in \url{https://github.com/fafafafafafafa7/LS_Algorithm}\label{code_site}} and compiled with g++ version 5.4.0 with -O3 option. These experiments are carried out under Ubuntu 16.04.3 LTS, using an Intel Core i5-7200U CPU at 2.50GHZ and 8GB RAM. For algorithms involving randomness, we run them for $10$ times and report the average RF.



\noindent{\textbf{Datasets.}}  We use real datasets from the Network Data Repository online~\cite{DBLP:conf/aaai/2015} which is a well-known network repository containing a large number of networks in several different domains.
To make an extensive evaluation of our algorithms, we select 1872 graphs from that repository after discarding datasets that are not in any graph format or of small sizes (less than 1,000 edges) since it is less interesting to partition small graphs in a distribution system.
We ensure that these selected graphs can cover both a wide range of size levels (from thousand edges to more than 17 million edges) and various domains (including 19 different domains: from real-life social graphs to manually generated graphs).
The detailed information about these 1872 selected graphs can be found in our GitHub repository\textsuperscript{\ref{code_site}}, which also contains our source code and some detailed experimental results.

\noindent{\textbf{Average Evaluations.}}
First, we present the average results on the 1872 datasets with the default setting $k=64$ and $\alpha=1.1$ to show the big picture of the performance.
Overall, we can get an average improvement of $12.07\%$ for LS-G and $13.20\%$ for LS-F. The detailed improvements with the four different initial partitions are shown in Fig.~\ref{f-4} and Fig.~\ref{f-5}, where we also show the proportion of the instances with the best result for four different initial partitions.

On average, LS-F can get more improvements compared with LS-G. However, LS-G will use less running time. The comparison of running time is omitted here due to the limited space.
The improvements on NE are not significant. However, for most instances, the performance of METIS+LS-G and METIS+LS-F is much better than that of NE+LS-G and NE+LS-F.
For LS-F, about $42.51\%$ best results are obtained by using initial partitions of METIS.
\begin{figure}[h!]
\begin{minipage}[t]{0.5\linewidth}
\centering
     \setlength{\belowcaptionskip}{-0.5cm}  
\includegraphics[width = 1\linewidth]{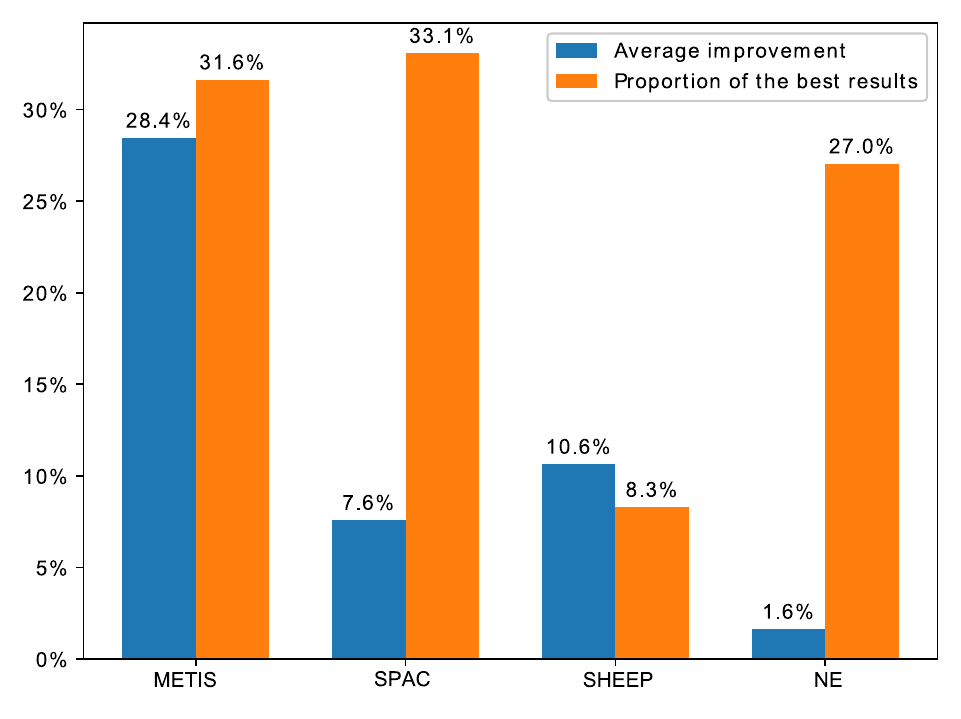}
\caption{LS-G.}\label{f-4}
\end{minipage}%
\begin{minipage}[t]{0.5\linewidth}
\centering
\includegraphics[width = 1\linewidth]{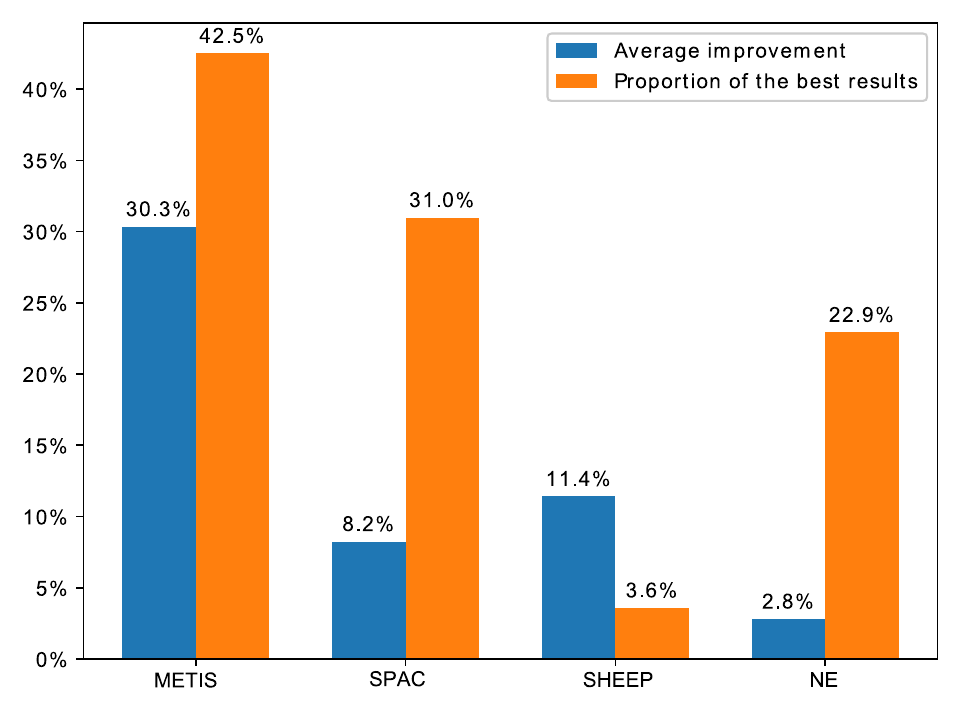}
\caption{LS-F.}\label{f-5}
\end{minipage}
\end{figure}

\noindent{\textbf{Detailed Comparisons.}}
To give clear comparisons, we also select four concrete instances from the 1872 instances as examples to illustrate the details.
The four instances are selected from four different domains with different size levels: coauthors-dblp $(540486, 15245729)$ from ``collaboration networks'', grid-yeast $(6008, 156945)$ from ``biological networks'', Texas84 $(36364, 1590651)$ from ``Facebook networks'', and lastfm $(1191805, 4519330)$ from ``social networks''. The two numbers in the brackets are the numbers of vertices and edges of the graph.

Most previous algorithms, say METIS, NE, SHEEP, and SPAC, fixed the balance value $\alpha=1.1$. So we also take this setting and show the results under different values of $k$ in Figure~\ref{figGF}.
We can see that as $k$ grows, LS-G and LS-F can consistently and effectively enhance the results produced by initial algorithms.

\vspace{-3mm}

\begin{figure}[h!]
    \centering
         \setlength{\belowcaptionskip}{-0.2cm}  
    \includegraphics[width = 0.8\linewidth]{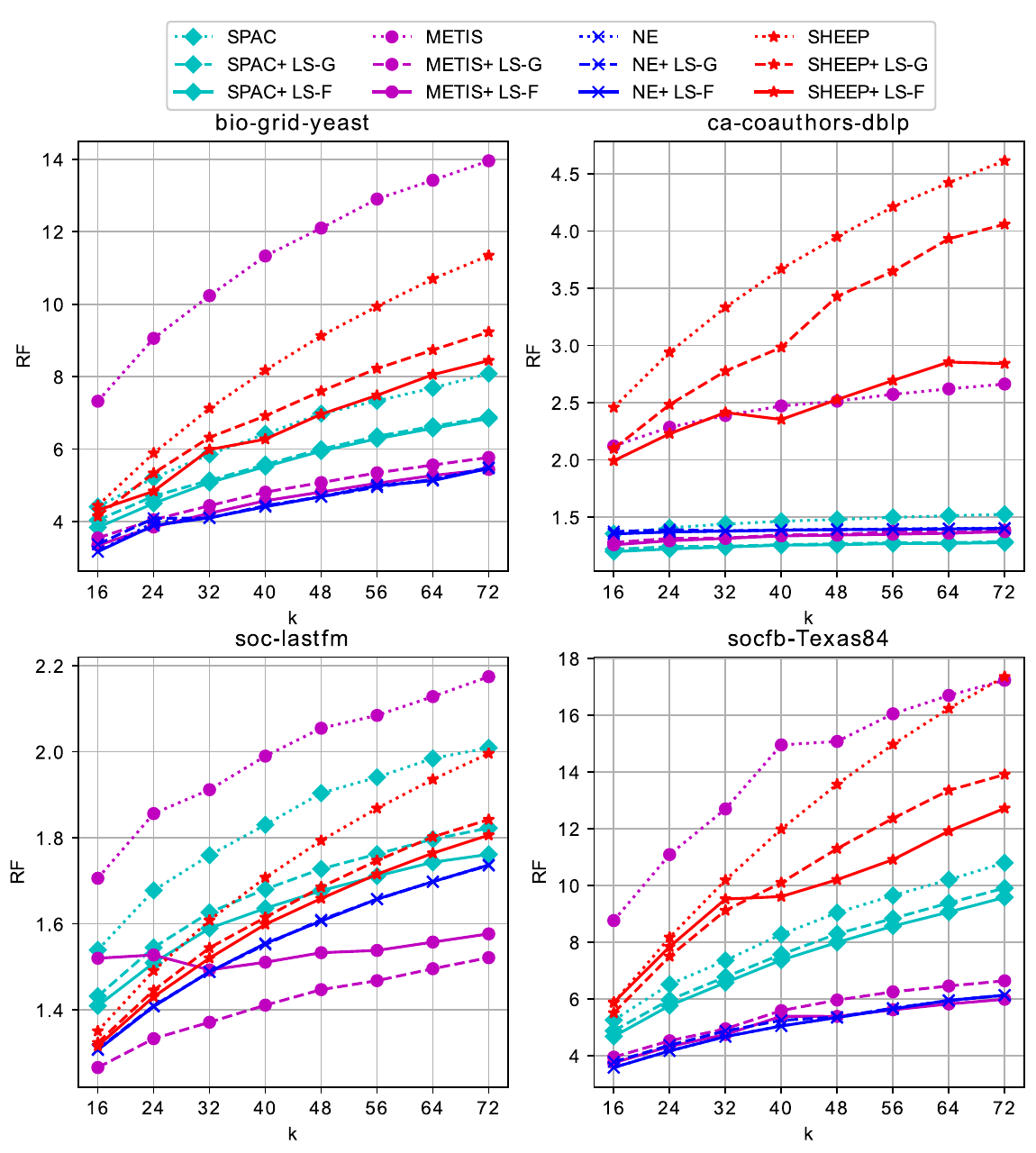}
    \caption{The results with different values of $k$.}\label{figGF}
\end{figure}

To make a full understanding of our local search methods,
we also do break-down analysis by visualizing the effect of our local search methods.
We show in Table~\ref{blocknum} the number of blocks in the initial partitions generated by different edge partitioners.
We can see that METIS generates much more blocks than the other three algorithms, which as a consequence enlarges the search space of METIS.
The number of blocks generated by NE is small. So for initial partitions generated by NE, our local search algorithms can only improve a small part, as shown in Fig.~\ref{f-4} and Fig.~\ref{f-5}.

\begin{table}[h!]
    \centering
         \setlength{\belowcaptionskip}{-0.2cm}
    \caption{The number of blocks.}\label{blocknum}
    \small
    \begin{tabular}{|l|c|c|c|c|}
    \hline
    & grid-yeast & coauthors-dblp & lastfm & Texas84   \\ \hline
    METIS   & 76,321  & 870,709    & 1,172,995 & 587,398 \\ \hline
    SHEEP   & 40,107  & 567,944    & 226,445   & 477,879 \\ \hline
    SPAC    & 10,997  & 139,119    & 395,783   & 64,399  \\ \hline
    NE      & 642     & 1534       & 238       & 127    \\ \hline
    \end{tabular}
\end{table}

Figure~\ref{figBloNum} shows the number of blocks of different sizes before and after applying LS-G and LS-F on METIS.
More than $95\%$ blocks are of size at most 20.
Most of the blocks with medium or large sizes have been removed by our algorithms.
 For these instances, LS-F can reduce more blocks and it generates a better result than LS-G. We can also see that the number of size-$1$ blocks drops sharply. The number of adjustable edges incident on size-1 blocks may be small and then it may
be easy to be reduced by our algorithms. So size-1 blocks should contribute the improvement greatly.

\vspace{-2mm}

\begin{figure}[h!]
    \centering
         \setlength{\belowcaptionskip}{-0.6cm}  
    \includegraphics[width = 1\linewidth]{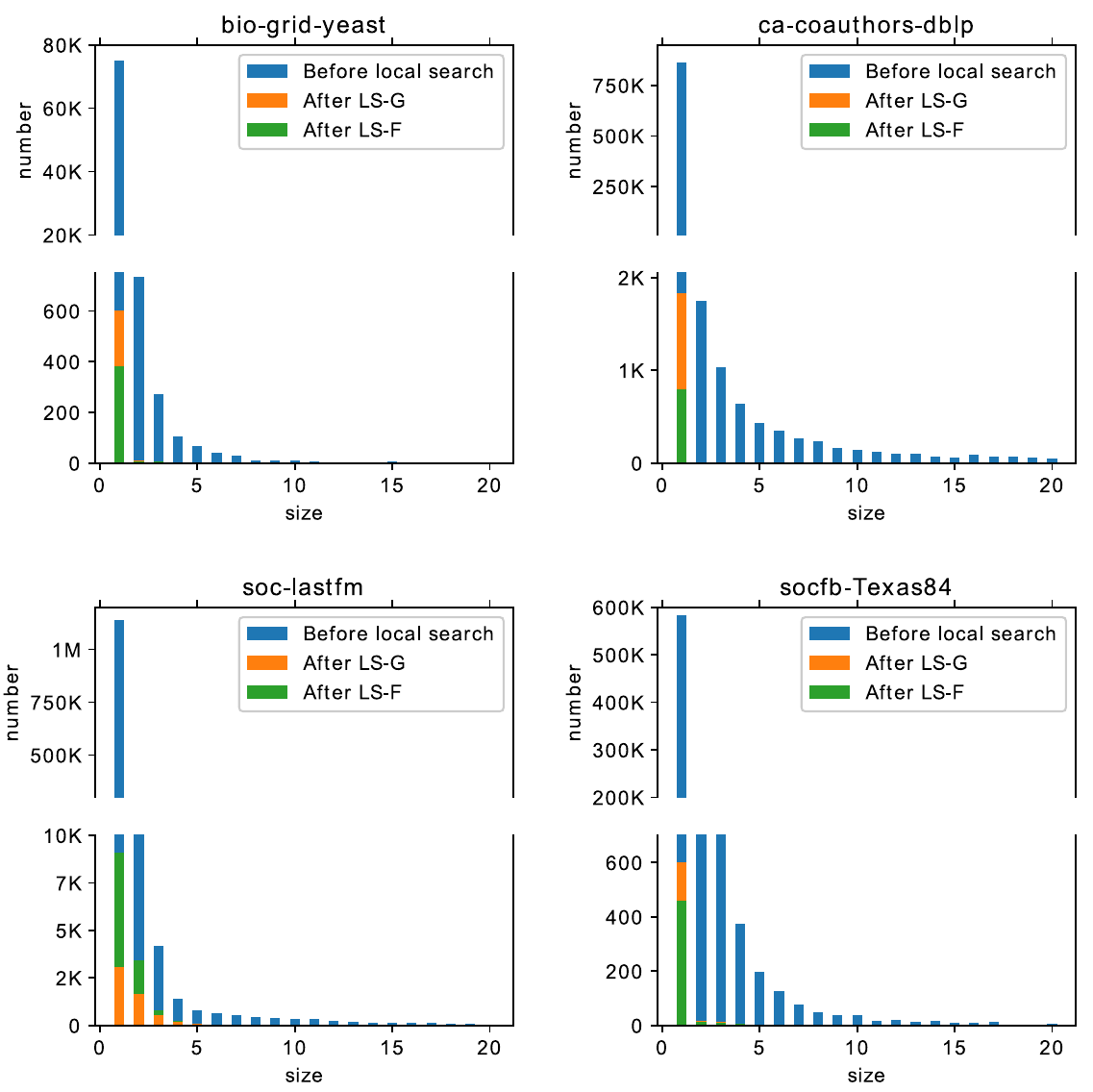}
    \caption{The number of blocks before and after applying LS-G and LS-F on METIS.}\label{figBloNum}
\end{figure}

\subsection{Further Applications and Discussion}
 There are two frequently used computation tasks in graph-parallel computation: PageRank~\cite{BRIN1998107} and triangle counting~\cite{zhang2017graph,xie2014distributed}. We run these two tasks on the above four selected instances.
Compared with initial edge partitions obtained by METIS, SHEEP, SPAC, and NE, our local search algorithms improve the running time for most cases (except NE on grid-yeast) with an average speedup of $9.23\%$ for LS-G
(resp., $10.60\%$ for LG-F) in task PageRank; and with an average speedup of  $8.00\%$ for LS-G (resp., $7.92\%$ for LG-F) in task triangle counting. The detailed results can be found in the Appendix.
 Note that although our algorithms always get improved edge partitions, the total running time of concrete computation tasks may not always be improved. The reason should be that the different computation tasks within some components may become worse even the partition is improved.

GraphX is a well known graph-parallel computation system~\cite{gonzalez2014graphx}. It has some built-in edge partition algorithms, which cannot be exported from the system, and then we
are unable to apply our local search algorithms on them directly. Compared with the results obtained by the built-in algorithms of GraphX, our local search algorithms (with METIS, SHEEP, SPAC, and NE) can get an average speedup of $30.82\%$ (the running time of computation tasks after giving the partition).
Further applications of our algorithms in graph analytic systems are worthy of deep study. Anyway, this paper gives some structural properties and efficient algorithms together with theoretically proved approximation ratios for an important optimization problem.



\section*{Acknowledgements}
The work is supported by the National Natural Science Foundation of China, under grants 61972070 and 61802049,
Singapore Ministry of Education, under grant MOE2017-T2-1-141, and Sub Project of Independent Scientific Research Project, under grant ZZKY-ZX-03-02-04.

\bibliographystyle{aaai21}
\bibliography{ep}  


\end{document}